%% file: main.tex
\documentclass[copyright,creativecommons]{eptcs}
\providecommand{\event}{WORDS 2011} 

\usepackage{amsthm}
\usepackage{amssymb}
\usepackage{amsmath}
\usepackage{mathrsfs}
\usepackage{times}
\usepackage{latexsym}
\usepackage{hhline}
\usepackage{enumerate} 
\usepackage{slashbox}
\usepackage[latin1]{inputenc} 
\usepackage{calc}         
\usepackage{graphicx}     
\usepackage{ifthen}       
\usepackage{pst-all}      
\usepackage{pst-poly}     
\usepackage{multido}      

\title{From Regular to Strictly Locally Testable Languages
\thanks{Extended Abstract. 
}}

\author{Stefano Crespi Reghizzi
\institute{Dipartimento di Elettronica e Informazione \\ Politecnico di Milano} 
\email{crespi@elet.polimi.it}
 \and Pierluigi San Pietro
\institute{Dipartimento di Elettronica e Informazione \\ Politecnico di Milano} 
\email{sanpietro@elet.polimi.it}
}

\def\titlerunning{From Regular to Strictly Locally Testable Languages}
\def\authorrunning{S. Crespi Reghizzi and P. San Pietro}
\newtheorem{definition}{Definition}
\newtheorem{example}{Example}
\newtheorem{proposition}{Proposition}
\newtheorem{lemma}{Lemma}
\newtheorem{theorem}{Theorem}
\newtheorem{remark}{Remark}

\begin{document}

\maketitle

\begin{abstract}
A classical result (often credited to Y. Medvedev)  states that  every  language recognized by a finite automaton is the homomorphic image of a local language, over a much larger so-called local alphabet, namely the alphabet of the edges
of the transition graph. Local languages are characterized by the value $k=2$ of the sliding window width in the  McNaughton and Papert's infinite hierarchy of
strictly locally testable languages ($k$-slt). We generalize Medvedev's result in a new direction, studying the relationship between the width and the alphabetic ratio telling how much larger the local alphabet is. 
We prove that every regular language is the image of a $k$-slt language on an alphabet of doubled size, where the width logarithmically depends on the automaton size,
and we  exhibit regular languages for which any smaller alphabetic ratio is insufficient. More generally, we express the trade-off between alphabetic ratio and width  as a mathematical relation derived from a careful encoding of the states. 
At last we mention some directions for theoretical development and application.
\end{abstract}

\input{intro}
\input{results}

\input{mainRes}

\section{Conclusion}\label{SectConclusion}
We have generalized Medvedev's homomorphic characterization of regular languages:  
instead of using as generator a local language over a large alphabet, 
which depends on the complexity of the regular language, we can use a strictly locally testable language over a smaller alphabet that does not depend on complexity, but just on the source alphabet. 
We have proved that the smallest alphabet one can use in the generator is the double of the alphabet of the regular language; thus, for instance,  four symbols suffice to  homomorphically generate any regular binary language. 
\par
In the main proof we have offered a specific and fairly optimized construction of the strictly locally testable language, for which  we have derived the relationship between the width,   the alphabetic ratio, 
and the complexity of the regular language.  In our opinion, the construction  should be of its own interest, as a new technique for simulating a NFA by means of a larger, yet strictly locally testable, machine. 
Our encoding is asymptotically optimal with respect to language complexity, and remains very close to the theoretical optimum for finite values of complexity.
But it is an open technical question whether a different construction  would yield better values  for the alphabetic ratio and the width parameter.
\par
Applications and developments of our result are conceivable in areas where a language characterization \textit{\`{a} la} Medvedev has been found valuable, as in the next ones.
\\
Picture languages. A main family of 2-dimensional languages, the tiling systems \cite{GiammRestivo1997}, is defined by a 2-dimensional Medvedev characterization. Does our result extend to 2D languages?
\\
Context-free languages. Combining our result with the Chomsky-Schutzenberger theorem it should be possible to obtain non-erasing homomorphic characterizations using a small alphabet.
\\
Consensual languages \cite{CrespiSpietro2001}. This generalization of finite-state machines motivated by modelling tightly connected concurrent computations uses homomorphism between words as its core mechanism.
\\
Information transmission for reducing the receiver cost was already mentioned in the introduction.
\par
{\textit{Acknowledgments:} } Thanks to Aldo De Luca for suggesting relevant references.
\bibliography{FormalLangBib}
\bibliographystyle{eptcs}
\end{document}

%% file: intro.tex
\section{Introduction}\label{SectIntroduction}
A classical result \cite{Eilenberg74}, often credited to Y. Medvedev \cite{Medvedev1964}, 
states that every regular language is the homomorphic image of a local language over a
larger alphabet called \textit{local}. In a local language the
sentences are characterized by three sets: the initial letters,
the final letters and  the set of factors of length $k=2$.  Parameter 
$k$ is  the \textit{width} of the simplest sliding window device introduced by McNaughton and Papert \cite{McPa71}. The
result simply derives from the fact that the set of  paths in an
edge-labelled graph is a local language over the alphabet of the
edges. Considering  a finite automaton for the regular language,
the local language of accepting paths  can be naturally projected
on the original language.
\par
Our work originates from two observations. First, in the classic
result the alphabet  of the local language is
larger than the source alphabet, by a multiplicative factor, to be called the \textit{alphabetic ratio},  in the order of the square
of the number of states. 
The simplicity of sliding window machines and languages is very attractive, but the huge size of the local alphabet in Medvedev theorem makes 
 their application impractical. 
\par
Then a
natural question concerns the local alphabet in the
classical result: how small can the alphabetic ratio be? A small alphabet may, for instance, allow to encode messages from a regular language 
into an slt language, to be transmitted over a communication channel, 
so that a more economical sliding window receiver can be used instead of a general finite state machine. 
\par
\par
Second, the local languages are a member of
 McNaughton and Papert's \cite{McPa71}
infinite hierarchy of  $k$-\textit{strictly locally testable}, for short $k$-slt, languages.
Then, by considering  $k$-slt, instead of just 2-slt i.e., local languages,
we raise a more general question: what is the minimum alphabetic ratio such that, for some finite parameter $k$, every
regular language is the alphabetic homomorphism of a $k$-slt
language? In that case, how big does the width parameter $k$ need to be?
More precisely, our main result, which  generalizes Medvedev theorem, expresses the trade-off between two parameters: the alphabetic ratio and the width.
\par
We spend a few lines to show the early but enduring interest for subfamilies of regular languages characterized by some form of local testability, without entering into details.
\par 
At the basis of formal language theory, the classical theorem of N. Chomsky and  M.P. Schutzenberger  characterizes context-free languages by a homomorphism applied to the intersection of a Dyck language and a 2-slt one. Several similar characterizations for other language families have later been proved.
In mathematics, the slt languages have been applied in the theory of  semigroups by  A. De Luca and A. Restivo \cite{LucRes80}.
In linguistics, a persistent idea is that natural languages can be modeled, at various levels, by locally testable properties. For instance, the psychologist W. Wickelgren \cite{Wickelgren1969}
made the observation that the set of English words are essentially a 3-slt (finite) language, and several brain scientists (in particular V. Braitenberg \cite{Braitenberg2004}) have suggested that sequences of finite length, such as the factors occurring in a locally testable language, can be easily stored and recognized by certain neural circuits (in particular the synfire chains of  M. Abeles) that have been observed in the cortex.
In computational linguistics  locally testable definitions have proved to be useful at various levels of finite-state models. 
Many persons (e.g. \cite{GarciaVidal1990}) working on language learning models have been attracted by the efficiency of  learning algorithms for various types of locally testable languages. Contemporary comparative work on the aural pattern recognition cababilities of humans and animals \cite{RogersPullum2011} have called attention to the subregular hierarchies induced by local testability.
In mathematical biology, in his seminal article on language theory and DNA \cite{Head1987}, T. Head shows that certain splicing languages are precisely the slt languages.
\par
The paper is organized as follows. After the basic definitions in
Section \ref{SectDefinitions}, we introduce in Section
\ref{SectLowerBounds} a new classification of regular
languages based on their homomorphic characterization via a $k$-slt language over an alphabet of size $m$.
In Section \ref{SectLowerBounds} we prove a lower bound on  the alphabetic ratio. In Section \ref{SectMainResult} we state and demonstrate a generalization of Medvedev theorem, 
including a mathematical analysis of the relationship between language complexity, alphabetic ratio, and width.   The  Conclusion presents an open problem and mentions conceivable developments and applications of the main result.

\section{Preliminaries}\label{SectDefinitions}
The empty word is
denoted by $\varepsilon$. The terminal alphabet of the source language  is denoted by $A$. For simplicity  we deal only with languages in $A^+$, which do not
contain the empty word. The cardinality of an alphabet will be 
called the \textit{arity}; the arity of a language is the arity of its alphabet.
\par
 A nondeterministic finite automaton
(NFA) $M$ is a quintuple $M = (Q,A,E,q_0,F)$ where $Q$ is a finite
set of states, $A$ is a finite alphabet, the transition relation (or graph) is $E \subseteq Q \times A
\times Q$, 
$q_0 \in Q$ is the initial state; $F\subseteq Q$ is the set of final
states,  which does not contain $q_0$ (since only $\epsilon$-free languages are considered).
\\
Two transitions $(p,a,q)$ and $(p',a',q')$ are {\em consecutive} if $q=p'$. 
A {\em path} $\eta=e_0 e_1 \ldots e_{n-1}$  is a finite sequence of $n>0$ consecutive transitions 
$e_0 = (p_0, a_0,p_1)$, 
$e_1 = (p_1, a_1,p_2)$, $\dots$, 
$e_{n-1}=(p_{n-1}, a_{n-1},p_n)$.  
The {\em origin} of $\eta$ is $o(\eta)=p_0$, its {\em end} is $e(\eta)=p_n$, and its {\em label} is $l(\eta)=a_0 a_1 \dots a_{n-1}$. 
A {\em successful} path 
is a path with origin $q_0$ and end in $F$. The language recognized by $M$, denoted $L(M)$,
is the set of labels of all successful paths of $M$. 
\par
We assume, without loss of generality, that the transition relation is \textit{total}, i.e.,  
for every $q\in Q, a \in A$, set $\{p \in Q \mid (q,a,p)\in E\}\neq \emptyset$ 
(if $E$ is not total, just add a new sink state to $Q$).  
\par
Given another finite alphabet $B$, an \textit{(alphabetic)
homomorphism} 
is a mapping $\pi: B \to A$.
For a language $L'\subseteq B^{+}$,  its \textit{(homomorphic) image} under $\pi$ is the language
$L= \{\pi(x) \mid x \in L'\}$.
\par
For every word $w\in A^+$, for every $k\ge 2$, let $i_k(w)$ and
$t_k(w)$ denote the prefix and, respectively, the suffix of $w$ of
length $k$ if $|w|\ge k$, or $w$ itself if $|w|<k$. Let $f_k(w)$
denote the set of factors of $w$ of length $k$.  
Extend $i_k, t_k, f_k$ to languages as usual, i.e., $i_k(L) =
\{i_k(w) \mid w \in L\}$, $t_k(L) = \{t_k(w) \mid w \in L\}$, and
$f_k(L) = \bigcup_{w\in L} f_k(w)$. 
A  factor of a word $w$ starting at position $k$ and ending at position $h$, with $1\leq h,k \le |w|$, is defined as follows:
\[
\left\{ \begin{array}{ll}s_{k,h}(w)=\epsilon & \text{if } h<k \\
				s_{k,h}(w) = i_{h-k+1}\left(t_{|w|-k+1}(w)\right)	&	\text{otherwise }					
\end{array}
\right .
\]
Hence, for $h\ge k$, $|s_{k,h}(w)|= h-k+1$.
\begin{definition}\label{k-s.t.1}
A language $L$ is $k$-\textit{strictly locally testable}\footnote{The original name in
\cite{McPa71} is ``$k$-testable in the strict sense''. This concept should not be confused with other language families based on local tests, see \cite{Caron2000} for a recent account.}, shortly $k$-slt, 
if there exist finite sets $I_{k-1},T_{k-1}\subseteq A^{k-1}$ and $F_{k}\subseteq A^{k}$ such that, for every $x \in A^kA^*$,  the following  condition holds:
\begin{equation*}
\label{primaDef}
x\in L \iff	i_{k-1}(x)\in I_{k-1} \wedge
t_{k-1}(x)\in T_{k-1} \wedge
f_{k}(x)\subseteq F_{k}
\end{equation*}
A language is \textit{strictly locally testable} (slt) if it is $k$-slt for some $k$ to be called the {\em width}.
\end{definition}
 This definition ignores words  shorter than $k-1$, which however can be checked directly against a finite set, if needed.
The case $k=2$ corresponds to the very well known family of \textit{local
languages} (see for instance \cite{Eilenberg74} or
\cite{BerstelPin1996}).
 The following example will be referred to later.
\begin{example}\label{ex-2local}
 The language $L' = (a'a)^+ \cup (b'b)^+$ 
is $2$-slt, i.e., local, since it can be defined by the sets
$I_1 = \{a',b'\}$, $T_1 = \{a,b\}$,
$F_2 = \{a'a, b'b, aa', bb'\}$. 
\end{example}
It is known and straightforward to prove that the family of slt languages is strictly included in the family of regular languages, and it is an infinite strict hierarchy ordered by the width value.
For instance, the language $L_h = (ab^h)^+$ on
$A=\{a,b\}$, with $h>1$  a constant, is $(h+1)$-slt, but it is
not $h$-slt. In fact, $L_h$ is defined by the sets: $I_h = \{ab^{h-1}\}$, $T_h = \{b^h\}$, $F_{h+1}
= \{b^i a b^{h-i}\mid 0\le i \le h\}$. 
However, 
$L_h$ is not $h$-slt: consider the words $ab^h\in L_h$ and $ab^{h+1}\not\in L_h$:
$i_{h-1}(ab^h) = i_{h-1}(ab^{h+1})= ab^{h-2}$,
$t_{h-1} (ab^h) = t_{h-1}(ab^{h+1}) = b^{h-1}$, 
$f_{h}(ab^h) =  \{a b^{h-1}, b^h \} = f_{h}(ab^{h+1})$. 
Hence, the two words above cannot be distinguished by using width $h$. 

%% file: results.tex
\section{Lower Bounds}\label{SectLowerBounds}
As said, every regular language, to be referred to as \textit{source}, is
the  image of a 2-slt language whose arity may be
much larger than the arity of the source. To talk precisely about the width of the slt language and of the  ratio of the arities of the slt and source languages, we introduce a definition.
\begin{definition}\label{defImageOfSlt}
For $k\ge 2, m\ge 1$, a language $L\subseteq A^+$ is $(m,k)$-\emph{homomorphic}  if there exist an alphabet $B$ (called \textit{local}) of arity  $m$, a $k$-slt language $L'\subseteq B^+$,  and a homomorphism $\pi: B \to A$ such that  $L=\pi(L')$.
\end{definition}
Clearly, if $L\subseteq A^+$ is $k$-slt then $L$ is trivially $(|A|,k)$-homomorphic. Otherwise,
a local alphabet larger than $A$ is needed.  
For instance, the language $L = (aa)^+ \cup (bb)^+$ is not slt
but the language $L' = (a'a)^+ \cup (b'b)^+$ of Ex.~\ref{ex-2local} is $2$-slt.
By defining $\pi:\{a,a',b,b'\}\to \{a,b\}$ as $\pi(a)=\pi(a')=a$, $\pi(b)=\pi(b')=b$,
then $L=\pi(L')$ and hence $L$ is $(4,2)$-homomorphic. The alphabetic ratio of $L'$ and $L$ is $4/2=2$.
\par
The traditional construction (e.g. in \cite{Eilenberg74}) of a 2-slt language $L'$ considers an NFA $(Q,A,E,I,F)$ of size $n =|Q|$ for $L$, 
and uses  set $E$  as local alphabet, i.e., up to $n^2 \cdot|A|$ elements.
Hence we can restate Medvedev's property saying that every regular language on $A$ is $(n^2\cdot|A|,2)$-homomorphic (the alphabetic ratio is $n^2$).
However, it is straightforward to show that the arity of the local alphabet can be reduced to $n\cdot |A|$.
\begin{proposition}\label{prop-2local}
Every regular language, accepted by an NFA with $n$ states, is $(n\cdot |A|,2)$-homomorphic.
\end{proposition}
\begin{proof}
Let $M=(Q,A,E,q_0,F)$ be an NFA. 
Define two mappings $\pi : Q \times A\to A$ and $\rho: Q \times A \times Q \to Q \times A $
such that $\pi(\langle q, a\rangle) = a$,  for every $a\in A$, $q\in Q$ 
and $\rho(p,a,q) = \langle p,a\rangle$ for every $p,q\in Q, a\in A$. 
The following sets define a 2-slt language $L'\subseteq \left(Q \times A\right)^+$:
\begin{align*}
I_1 & =\{\langle q_0,a\rangle \mid a\in A\};\\
F_2 & =\{\langle q,a\rangle \langle q',b\rangle \mid a,b\in A, q,q'\in Q, (q,a,q') \in E\};\\
T_1 & =\{\langle q,a\rangle \mid a\in A, \exists q'\in F: (q,a,q') \in E\}.
\end{align*}
We show first that
$\pi(L') \subseteq L$. Let $w \in \pi(L')$.
Hence, there exists $x\in L'$ such that $\pi(x) = w$. 
We claim that there exists a successful path $\eta$ of $M$ such that $x = \rho(\eta)$. 
Let $n = |w|$. Since $x \in L'$, there exist $q_1,q_2, \dots q_{n-1} \in Q$, $a_0, a_1, \dots, a_{n-1} \in A$ such that 
$x = \langle q_0, a_0 \rangle  \langle q_1, a_1 \rangle \dots \langle q_{n-1}, a_n \rangle$, and $w = a_0 a_1 \dots a_{n-1}$. 
Since $\langle q_{n-1}, a_n \rangle \in T_1$, there exists $q\in F$ such that $(q_{n-1}, a_{n-1},q) \in E$.
Let $\eta$ be $(q_0, a_0,q_1)$ $ (q_1, a_1,q_2)$ $ \dots ( q_{n-1}, a_{n-1},q)$: $\eta$ has label $w$, origin in $q_0$ and end in a final state; 
moreover, $\rho(\eta)=x$.
By definition of $F_2$, every factor 
$\langle q_{i-1},a_i\rangle \langle q_{i},a_{i+1}\rangle$ of $x$, for $1\le i\le n$, must be such that $(q_{i-1},a_i,q_i)\in E$, hence all transitions of $\eta$ are consecutive, i.e., $\eta$ is a successful path of label $w$.
\\
We show that $L \subseteq \pi(L')$. 
Let $w \in L$ be accepted by a successful path $\eta$ of $M$ of the form  
\[(q_0, a_0,q_1) (q_1, a_1,q_2) \dots ( q_{n-1}, a_{n-1},q_n),\]
 with $q_n \in F$ and  $a_0 \dots a_{n-1} = w$. 
We claim that $\rho(\eta) \in L'$. In fact, 
$i_1(\rho(\eta)) = \langle q_0, a_0 \rangle \in I_1$, 
$t_1(\rho(\eta)) = \langle q_{n-1},a_n \rangle \in T_1$ and 
$f_2(\rho(\eta)) = \{ \langle q_{i-1},a_{i-1}\rangle \langle q_i,a_i\rangle \mid 1\le i \le n\}$. 
Since each $(q_{i-1} a_{i-1} q_i) \in E$ (being a transition of $\eta$), $f_2(\rho(\eta)) \subseteq F_2$.
\end{proof}
A natural question to be later addressed, is whether, by allowing the width $k$ to be larger than 2, 
it is possible to reduce the arity of the local alphabet to less than $n\cdot |A|$. 
Next we prove the simple, but perhaps unexpected result, that the local alphabet cannot be smaller than twice the size of the source one.
\begin{theorem}\label{theorem1}
For every alphabet $A$, there exists a regular language $L\subseteq A^+$ that is not $\left(2\cdot|A|-1,k\right)$-homomorphic, for every $k\ge 2$.
\end{theorem}
\begin{proof}
Let $L$ be defined by the regular expression $\bigcup_{a \in A} (aa)^*$.
By contradiction, assume that there exist $k\ge 2$ and a local alphabet $B$ of arity $2\dot|A|-1$, a mapping $\pi:B \to A$ and a $k$-slt language $L'\subseteq B^+$ such that $\pi(L')=L$.
Since $|B| = 2\cdot|A|-1$, there exists at least one symbol of $A$, say, $a$, such that there is only one symbol $b\in B$ such that $\pi(b) = a$.
Since the word $a^{2k}\in L$, there exists $x\in L'$ such that $\pi(x)=a^{2k}$. By definition of $\pi$ and of $B$,
$x = b^{2k}$.  Consider the word $x b = b^{2k+1}$.
Clearly, $\pi(x b)= a^{2k+1}$, which is not in $L$, since all words in $L$ have even length.
Hence, $x b\not\in L'$. But $i_{k-1}(x) = i_{k-1}(xb)= b^{k-1}$, $t_{k-1}(x)  = t_{k-1}(xb) =  b^{k-1}$, $f_{k}(x) = f_{k}(xb)=
 b^{k}$ and, by Definition \ref{k-s.t.1}, $x b$ is in $L'$, a contradiction.
\end{proof}
The same result holds (with a very similar proof) if in the statement the class of \textit{strictly} locally testable languages is replaced by the class of \textit{locally testable} languages\footnote{They are the boolean closure of slt languages, see \cite{McPa71}.}.
The  question whether an alphabetic ratio of two is sufficient is addressed in the next section.  

%% file: mainRes.tex
\section{Main Result}\label{SectMainResult}
The intuitive idea that by increasing the width  one can use a smaller alphabet for the slt language, is studied in detail. 
Our approach  consists of defining an slt language using a larger alphabet 
that encodes the  states traversed by the original automaton into words of fixed length.  
Our main theorem states the relationship between the language complexity in terms of number of states, the alphabetic ratio, and the width of the slt language.
\begin{theorem}\label{th-main}
If a language $L\subseteq A^+$ is accepted by a NFA with $n>1$ states, then
for every $h\ge2$, 
$L$ is $\left(h|A|,O(\frac{\lg{ n}}{\lg{h}})\right)$-homomorphic. 
\end{theorem}
The rest of the section is devoted to the proof. Special care is devoted to find a very succinct encoding  of the original states into strings of the local alphabet, in order to reach 
the minimal alphabetic ratio. 
Since it may be important for applications, our encoding produces also a small, although not optimal, width of the slt language.
The proofs are organized so that 
the main lemmas hold, independently of the chosen encoding, which only affects the numerical results. 
This organization has the advantage that the proof is essentially unaffected by the encoding. 
\par
The next definitions set the base for stating the properties a good encoding should have. Only fixed-length encodings are considered.
Let $D$ be a finite alphabet.
Let $M = (Q,A,E,q_0,F)$ be a NFA, where $E$ is total, and let 
$n = |Q|>1$. 
\par
Given an integer $m\ge \lceil lg_{|D|} (|Q|)\rceil$, a {\em code} of $Q$ into $D$ of length $m$ is a mapping $[\text{\;}]:Q \to
D^m$ such that for every  $p,q\in Q$, if $p\neq q$ then $[p] \neq [q]$.
Consider a word $x$ that is a factor of $[Q^+]$. We want to decode $x$ to one
state. This will be useful when defining a slt language whose homomorphic image is $L(M)$.   
If $|x|\ge 2m$, since $x$ may include the concatenation of $[q]$ and
$[p]$, $q,p\in Q$, it is not decodable to just one state symbol; moreover, if $|x|<2m-1$ then $x$
may not
contain any factor of the form $[q]$. However, if $|x|$ is exactly $2m-1$, then the word is
bound to include at least one factor of the form $[q]$, for some $q \in Q$, which can be
decoded to $q$.  In addition, we want this decoding to be unique. 
\par
The traditional notion of decodability 
 (for every $x,y\in Q^+$, if $[x]=[y]$ then $x=y$) is not adequate, since
it assumes that the word to be decoded is a string in $[Q^+]$, while we need to consider a {\em factor}
of $[Q^+]$.
A word $x\in D^{2m-1}$ is said to be
{\em factor-decodable} if there exists one, and only one, position $j$, $1\le j \le m-1$, 
such that there exists $q \in Q$: $s_{j,j+m}(x) = [q]$.
A code $[\text{\;}]:Q\to D^m$ is  {\em factor-decodable}  if every
word in $f_{2m-1}([Q^+])$ is factor-decodable. 

\begin{lemma}\label{factor-lm}
For all finite alphabets $Q,D$ of cardinalities $n=|Q|$ and $h=|D|$, with $n\ge 2$, $2\le h< n$, there
exists a factor-decodable code of $Q$ into $D$ of length $m=\lceil g(h) +  f(h)\lg_2{n}\rceil\ge 3$,
with: 
\begin{align*}
f(h) &= \lg^{-1}_2 {\left(h-1+\sqrt{(h-1)(h+3)}\right)} -1 \\
g(h) &= 1 + \frac{f(h)}{2}\left(\lg_2(h-1)+ \lg_2(h+3)\right).
\end{align*}
\end{lemma}

\noindent{\bf Sketch of the proof.} Let $0\in D$ be a symbol. The idea is to 
let code $[\text{\;}]$ be such that for every $q \in Q$, $[q]$ ends with the word $00$, i.e., 
$s_{m-1,m}([q])=00$ and there is no other occurence of 00 in $[q]$. 
Formally, for every $i$, $1\le i
\le m-1$, if 
$s_{i,i+1}([q])=00$ then 
$i=m-1$. 
This is enough for factor-decodability. 
 To find how large $m$ must be as a function of $h$ and $n$, first consider, for
every $m\ge 2$, the set $S(m)$ of words in $D^m$ such that
$x \in S(m)$ if $x$ has suffix 00 and in $x$ there is no other occurrence of $00$. If
$|S(m)|\ge n$, then it is possible to assign a distinct word in $S(m)$ to every state of $Q$.
The definition of $S(m)$ is by induction on $m\ge 2$. 
$S(2) = \{00\}$, i.e., the only word in $S(2)$
is 00. $S(3) = \{d00 \mid d \in D-\{0\} \}$. 
Given sets $S(m-1), S(m-2)$, let $S(m)$ be: 
\[\left\{dy \mid d \in D-\{0\}, y\in  S(m-1) \right\} \cup \left\{0dx \mid  d \in D-\{0\}, y\in 
S(m-2) \right\}.\]
Hence, $|S(2)| = 1$, $|S(3)| = h-1$ and  
\[|S(m)| = (h-1)\left|S(m-1)\right| + (h-1)\left|S(m-2)\right|.\]
This recurrence relation is strictly connected to the so-called Lucas sequence $U_m(P,Q)$, where
 $P,Q$\footnote{Beware that $Q$ is not the set of states.} are integers (see, e.g, p. 395 of \cite{Dick19}):
$U_1(P,Q) = 1$, $U_2(P,Q) = P$, and for $m\ge 3$, $U_m(P,Q) = PU_{m-1}(P,Q) - Q U_{m-2}(P,Q)$. For $P=1, Q=-1$ this is just a Fibonacci sequence.
If $P^2 -4Q\ge 0$, a closed-form solution for every $m>0$ is $U_m(P,Q) = \frac{a^m -b^m}{a-b}$, where 
$a = \frac{P+\sqrt{P^2 -4Q}}{2}, b = \frac{P-\sqrt{P^2 -4Q}}{2}$.
With standard algebraic manipulations and by defining $f(h), g(h)$ as in the statement of the Lemma, one can derive that: 
$|S(m)|\ge n$ is satified if
$m=\lceil g(h) +f(h) \lg_2{n}\rceil\hspace{8cm}$.

\begin{remark}
Both  $f(h)$ and $g(h)$ are monotonically decreasing with $h$, although very slowly for large $h$, with $\lim_{h \to \infty}{f(h)\lg_2 h}  = 1$, $\lim_{h \to \infty}{ g(h)} = 2$. 
with, moreover,  $0<f(h)\lessapprox 1.44$, $2\le g(h) \lessapprox 4.11$. The expression for $m$ is $O\left(\frac{\lg n }{\lg h}\right)$. 
By definition of a code, $m$ cannot be smaller then $m_{min} = \lceil \frac{\lg_2 n }{\lg_2 h}\rceil$, i.e., $m$ is $\Omega\left(\frac{\lg n }{\lg h}\right)$, hence the code of Lemma~\ref{factor-lm}
is asymptotically optimal. In particular, the ratio $m/m_{min}$, where $m$ is computed by the above formula, is dominated by term $f(h) \lg_2 h\lessapprox 1.44$, which is very close to $1$ for $h\ge 3$. 
Hence, no encoding can significantly improve $f(h)$ (or $g(h)$), decreasing $m/m_{min}$. A few examples of approximated values for $f(h)$, $g(h)$, and $f(h) \lg_2 h$ are:

\begin{center}
\begin{small}
\begin{tabular}{|r||r|r|r|r|r|r|}
\hline$h$ & 2 & 3 & 4 & 10 & 100 & 1000\\
\hline$f(h)$&  1.44 & 0.68 & 0.52 & 0.29 & 0.15 & 0.10 \\
\hline $g(h)$ & 4.11 & 2.92 & 2.66 & 2.34 & 2.15 & 2.10\\
\hline $f(h)\lg_2 h$ & 1.44 & 1.09 & 1.04 & 1.00 & 1.00  & 1.00\\\hline
\end{tabular}\end{small}
\end{center}
\end{remark}

To prove Th.~\ref{th-main}, a few more definitions are required. 
Define the following alphabetic homomorphisms:
$\alpha: A \times D \to A$, $\delta: A \times D \to D$ are such that $\alpha(a,d) = a$, $\delta(a,d)=d$
for every $a\in A, d \in D$.
\\
A path of $M$ of length $t\ge 0$ is called a $t$-path. 
Paths $\eta_1, \eta_2, \dots, \eta_k$ of $M$, $k\ge 2$, are called {\em consecutive}  if 
$\eta_1 \eta_2 \dots \eta_k$ is also a path
of $M$ (i.e, $e(\eta_h) = o(\eta_{h+1})$, for all $1\le h \le k-1$).
With an abuse of notation, let $[\text{\;}]: (Q\times A \times Q)^* \to (A \times D)^*$ be defined on paths as follows. 
Let $\eta$ be a $t$-path. If $t=0$ then $[\eta]=\epsilon$; 
if $1\le t \le m$, let $[\eta]$ be the unique word $z$ in $(A \times D)^m$ such that $\alpha(z)= l(\eta)$, 
$\delta(z)= i_{t}([o(\eta)])$ (i.e., $\delta(z) = [o(\eta)]$ if $\eta$ is a $m$-path). 
\\
If $|\eta|>m$, then there exist a unique $k\ge 1$ and a unique $0\le j \le m-1$ such that $|\eta| = km+j$; 
hence, there exist $k+1$
consecutive paths of $M$, denoted by $\eta_1, \eta_2, \dots,  \eta_k, \eta_{k+1}$ such that 
$\eta = \eta_1 \eta_2 \dots  \eta_k \eta_{k+1}$, each $\eta_h$, $i\le h \le k$, is a $m$-path and 
$\eta_{k+1}$ is a $j$-path. This decomposition in
consecutive paths is called the {\em canonical decomposition} of $\eta$. 
Then, $[\eta]$ is defined as 
$[\eta_1] [\eta_2] \dots  [\eta_k] [\eta_{k+1}]$.
\par
Let $L'$ be the $2m$-slt language defined by the following sets:
\begin{eqnarray*}
 I_{2m-1} & = & i_{2m-1} \left(\{ [\eta' \eta''] \mid \eta', \eta'' \textrm{ are  consecutive $m$-paths of
} M \land  \delta([\eta'])=[q_0]\}\right);\\
 F_{2m}\;\;\; & = & f_{2m}\left(\{[\eta' \eta'' \eta''']\mid \eta', \eta'', \eta''' \textrm { are
consecutive $m$-paths of } M\}\right);\\
 T_{2m-1} &= & t_{2m-1} (\{[\eta' \eta'' \eta''']\mid  \eta', \eta'', \eta''' \textrm {  are
consecutive paths of } M, |\eta'|=|\eta''| = m, \\
& & 0\le |\eta'''|< m, e(\eta'' \eta''')\in F\}).
\end{eqnarray*}

The proof of the following lemma follows from uniqueness of factor-decodability:
\begin{lemma}\label{lm-decod}
Let $[\text{\;}]:Q\to D^m $ be a factor-decodable code. 
For all $z\in F_{2m}$, there exist a 
position $j$, $1\le j \le m-1$, 
and two consecutive paths $\eta_1, \eta_2$ of $M$ such that:  
\begin{enumerate}
\item $\eta_1$ is a $m$-path, 
and $t_{2m-j+1}(z) = [\eta_1] [\eta_2]$; 
\item for any two consecutive paths of $M$, $\eta_{I}, \eta_{II}$, 
if  $\eta_I$ is a $m$-path and $[\eta_I] [\eta_{II}]$ is a suffix of 
$z$ then $[\eta_1] = [\eta_{I}]$ and
$[\eta_2] =[\eta_{II}]$;
\item if $\delta(i_{m}(z))=[q]$ for some $q \in Q$, 
then $j=1$, $\eta_2$ is a $m$-path and $o(\eta_1) =q$;
\item if $t_{2m-1}(z) \in T_{2m-1}$, then $e(\eta_1 \eta_2) \in F$.
\end{enumerate}
\end{lemma}

\begin{lemma}\label{lm-slt}
There exists a finite language $L''\subseteq A^+$ such that $\alpha(L') \cup L'' = L(M)$. 
\end{lemma}

\noindent{\bf Sketch of the Proof} Let $L''$ be the  set of words in $L(M)$ of length less than $3m$. 
\par
Part (I): $(L(M)- L'')\subseteq \alpha(L(M')$. 
Assume that $x \in L(M), |x| \ge 3m$. To show that there exists a successful path $\eta$ of $M$ such that $l(\eta) = x$, we first claim the following result for every path, whether successful or not:
\[ 
(*) \text{ for all paths } \eta \text{ of } M, \text{ with } |\eta|\ge 3m, f_{2m}([\eta]) \subseteq F_{2m}.\]

The proof of (*) is on induction on the 
the canonical decomposition of $\eta$.
Part (I) can now be completed. For all $x \in L(M)- L''$, let $\eta$ be a successful path of $M$ 
with $l(\eta) = x$; moreover, let $\eta_1, \dots, \eta_k, \eta_{k+1}$ 
be the canonical decomposition of $\eta$. 
By (*), $f_{2m}([\eta]) \subseteq F_{2m}$. 
But $\eta$ is successful:   $o(\eta) = o(\eta_1) = q_0$, hence 
$i_{2m-1}([\eta])= i_{2m-1}([\eta_1][\eta_2]) \in I_{2m-1}$; 
$e(\eta)\in F$, hence $t_{2m-1}([\eta])=t_{2m-1}\left([\eta_{k-1} \eta_k \eta_{k+1}]\right) \in T_{2m-1}$.
Therefore, $[\eta] \in L'$.
\par
Part (II):  $\alpha(L') \subseteq L(M)$.  
The proof needs the assumption that code $[\text{\;}]$ is factor-decodable.
The following property can be proved by induction on $k\ge 2$, by applying  Lemma~\ref{lm-decod}: 
\begin{quotation}
\noindent (+) for all words $z \in (A\times D)^+$, $|z|\ge 2m$, if $f_{2m}(z)\subseteq F_{2m}$ and
$i_{2m-1}(z) \in I_{2m-1}$ 
then there exists
a path $\eta$ of $M$ such that $z = [\eta]$  and
$o(\eta)= q_0$. 
\end{quotation}
The proof of Part (II) follows from (+). 
In fact, if $x \in \alpha(L')$, with $|x|\ge 3m$, then there exists $z \in L'$ such that $x = \alpha(z)$. 
Since in this case $i_{2m-1}(z) \in I_{2m-1}$,
$f_{2m}(z) \subseteq F_{2m}$, $t_{2m-1}(z) \in T_{2m-1}$, by (+) there exists a path $\eta$ of
$M$ with origin in $q_0$ and such that $z = [\eta]$.
Let $\eta_1, \eta_2, \dots, \eta_k, \eta_{k+1}$ be the canonical decomposition of $\eta$,  
with $|\eta| = km+j$, $k\ge3$ and $0\le j\le m-1$
(hence $|\eta_{k+1}| = j$). 
Let $w = t_{2m-1}([\eta_{k-1}] [\eta_k][\eta_{k+1}])$ and consider 
$t_{2m-1}(z) = t_{2m-1}([\eta]) =$ $t_{2m-1} ([\eta_{k-1}] [\eta_k] [\eta_{k+1}]) = w$.
Apply Lemma~\ref{lm-decod}, Part (1), to $w\in F_{2m}, w \in T_{2m-1}$. 
Hence, there exist a position $h$ and consecutive $\eta', \eta''$, with
$\eta'$ a $m$-path, such that 
$t_{2m-h}(w) = [\eta'][\eta'']$. Since $[\eta_k] [\eta_{k+1}]$ (of length $m+j\le2m-1$) is also a suffix of $w$, 
by Part (2) of Lemma~\ref{lm-decod}, $[\eta_k] = [\eta'], [\eta_{k+1}]=[\eta'']$.
Since $o(\eta_k) = o(\eta')$, also paths $\eta_{k-1}, \eta'$ are consecutive. Hence, 
$z =  [\eta] = [\eta_1 \dots \eta_{k-1} \eta_k \eta_{k+1}] = [\eta_1 \dots \eta_{k-1}] [ \eta_k
\eta_{k+1}] = [\eta_1 \dots \eta_{k-1}]  [\eta' \eta''] = [\eta_1 \dots \eta_{k-1} \eta' \eta'']$.
Therefore,  path $\eta_1 \dots \eta_{k-1} \eta' \eta''$
has label $z$, origin $q_0$, and end $e(\eta' \eta'')$ in $F$, i.e., it is successful: $x \in L$.

The proof of Th.~\ref{th-main} is now immediate.
By Lemmas ~\ref{factor-lm} and~\ref{lm-slt}, 
$ m  =  \lceil g(h) + f(h)\lg_2 n\rceil$, 
and $L'$ is $2m$-slt. Hence, $L$ is $(2|A|,2\lceil g(h) + f(h)\lg_2 n\rceil)$-slt, with $2\lceil g(h) + f(h)\lg_2 n\rceil$
being $O\left(\frac{\lg{n}}{\lg{h}}\right)$. 
A few examples of width for various values of number $n$ of states and alphabetic ratio $h$ are shown here. 

\begin{center}\begin{small}
                                    
\begin{tabular}{|r||r|r|r|r|r|r|}\hline
 \backslashbox{$h$}{$n$} & 10 &  $10^3$& $10^6$ & $10^9$ & $10^{40}$ \\\hline\hline 
 2 & 18 & 38 & 66 & 94 & 392\\ \hline
 3  & 12 & 20 & 34 & 48& 190 \\ \hline
 4 & 10 & 16 & 28 & 38 & 144 \\ \hline
 10  &8 & 12 & 18 & 24 & 86\\ \hline
 100  & 6& 8 & 12 & 14 & 46 \\ \hline
 1000 & 6 & 8 & 10 & 12 & 32\\\hline
\end{tabular}
\end{small}\end{center}

%
%
Hence, by enlarging the local alphabet, a smaller width suffices to construct the slt language.
However, it is useless to take an alphabetic ratio $h\ge n$, since in this case one can use the simpler construction of Prop.~\ref{prop-2local}.
To finish, we note that for many regular languages one can obtain a homomorphic definition that uses lower values of alphabetic ratio and/or width than those obtained by the main theorem.